\documentclass[a4paper,12pt]{article}

\usepackage[utf8]{inputenc}
\usepackage[T1]{fontenc}
\usepackage{hyperref}
\usepackage{xurl}
\usepackage{graphicx}
\usepackage{multicol}
\usepackage{subcaption}
\usepackage{tabularx}
\usepackage{amsfonts}
\usepackage{amsthm}
\usepackage{amssymb}
\usepackage{amsmath}
\usepackage{dsfont}
\usepackage{thmtools}
\usepackage{algorithm}
\usepackage{algpseudocode}
\usepackage[colorinlistoftodos,prependcaption,textsize=tiny]{todonotes}
\usepackage{lineno}
\usepackage{sidenotes}
\usepackage{marginnote}
\usepackage{array}
\usepackage{multirow}
\usepackage{authblk}
\usepackage{soul}
\usepackage{enumitem}
\usepackage{booktabs} 

\newtheorem{theorem}{Theorem}
\newtheorem{lemma}{Lemma}
\newtheorem{corollary}{Corollary}

\newtheorem{remark}{Remark}

\DeclareMathOperator {\diag}{\text{diag}}

\usepackage[top=2cm, bottom=2cm, left=2cm, right=2cm]{geometry}

\title{A dynamic mechanism for prevalence of triangles in competitive networks}

\author[1,2,*]{M.N.~Mooij}
\author[3,4,5]{M.~Baudena}
\author[2,6]{A.S.~von~der~Heydt}
\author[3,7]{L.~Miele}
\author[1,2,**]{I.~Kryven}

\affil[1]{\small Mathematical Institute, Utrecht University, Utrecht, The Netherlands}
\affil[2]{\small Centre for Complex Systems Studies (CCSS), Utrecht University, Utrecht, The Netherlands}
\affil[3]{\small National Research Council of Italy, Institute of Atmospheric Sciences and Climate (CNR-ISAC), 10133 Torino, Italy}
\affil[4]{\small National Biodiversity Future Center (NBFC), 90133 Palermo, Italy}
\affil[5]{\small Copernicus Institute of Sustainable Development, Utrecht University, Utrecht, The Netherlands}
\affil[6]{\small Institute for Marine and Atmospheric Research (IMAU), Utrecht University, Utrecht, The Netherlands}
\affil[7]{\small INRAE, UR1115, Plantes et Systèmes de culture Horticoles (PSH), Site Agroparc
84914, France}

\affil[*]{m.n.mooij@uu.nl}
\affil[**]{i.v.kryven@uu.nl}

\date{}

\begin{document}

\maketitle

\section*{Abstract}
Triangles are abundant in real-world networks but rare in standard null models for sparse graphs.
Existing explanations typically rely on explicit triadic closure mechanisms or geometry‑based connection rules. 
We propose an alternative hypothesis: the frequent appearance of triangles may arise naturally from the requirement of dynamic stability that maintains coexistence of species in Lotka–Volterra systems with equal competitive interactions.
To evaluate this idea, we prove that, across all possible interaction graphs, coexistence is guaranteed whenever the coupling strength is below the reciprocal of the graph's maximum degree, and guaranteed not to occur when the coupling strength exceeds 1. This leaves a large gap that is unexplained by the graph degrees alone. We notice that the lower and upper bounds are achieved for star and complete graphs respectively and
to investigate further what structural properties of the interaction graph control the critical coupling within the gap, we optimise networks algorithmically while keeping the degree sequence fixed. We find that networks supporting stronger interaction strengths consistently exhibit higher clustering coefficients in several network models.
Moreover, in real-world grassland plant networks, we observe higher clustering and stronger stability than those expected from a configuration model with the same degree sequence. 
Our result suggests that triangles, and clustering in general, may emerge as a structural signature of stabilising competition.
\\ \\
{\noindent \bf Keywords:} Triadic closure, Stability in Complex Systems, Competitive Dynamics, Network Optimisation.

\section{Introduction}
In sufficiently sparse network models where edges are placed independently at random, triangles become increasingly rare as the network size grows~\cite{bollobas2001}. This contrasts with observations in many large real-world networks, which have more triangles than their corresponding null models with the same degree sequence, for example, see the discussion in~\cite{watts1998, boguna2021} and the references therein.
 
While triangles are unavoidable in dense networks due to the threshold and saturation effects~\cite{bollobas2001}, in sparse networks their abundance is typically attributed to two main mechanisms.
In social networks, this phenomenon is commonly explained by \emph{triadic closure} -- the increased probability that two vertices become connected when they have a common neighbour.  In particular, this is reflected in models with strong local clustering, such as the Watts–Strogatz small‑world model~\cite{watts1998}.
In geometric networks, where vertices have coordinates in a latent metric space, local clustering arises naturally from the \emph{triangle inequality}: if two vertices are both close in latent space to a third one, they are also likely to be close to each other and hence to induce a (topological) triangle~\cite{boguna2021}. 
Finally, it is worth mentioning that simply conditioning a generic random graph model, such as the Erdős–Rényi random graph or exponential random graph models, to contain many triangles typically produces highly regular and hence unrealistic structures as a consequence of measure degeneracy~\cite{chakraborty2026,chatterjee2013estimating}.

Many networks that feature clustering, however, do not fall under either of the above‑mentioned categories. 
In this paper we give a possible explanation for the surplus of triangles using asymptotic stability theory for dynamical systems with competitive interactions.
Namely, we consider a second-order system of ODEs governed by the Lotka-Volterra (LV) equations on a network with symmetric negative interactions. The interaction strength, referred to as the \emph{competitive pressure}, and the carrying capacity are identical for all species. To quantify stability, we determine the maximum competitive pressure (in absolute value) that the system can sustain while still admitting a coexistence equilibrium, that is, an equilibrium in which all species maintain strictly positive abundances in the large‑time limit. We identify the maximum competitive pressure using a bifurcation analysis and call it the \emph{critical coupling} of the system, with larger values indicating a more robust coexistence equilibrium.

The relation between the critical coupling and the structure of the network has been a topic of extensive research. Robert May argued that the critical coupling of systems with a weighted Erd\H{o}s--R\'enyi interaction matrix tends to zero as the system size grows~\cite{may1972will}. 
Gao et al.\ observed that a bifurcation at the critical coupling in large dimensional systems with random interactions is captured by a one-dimensional ‘effective’ manifold~\cite{gao2016universal}.
In our previous works, we connected the critical coupling to Katz centrality~\cite{mooij2024findinglargeindependentsets} and showed that this quantity is bounded away from zero when the maximum degree (or the maximum total weight around a vertex) remain finite~\cite{mooij2024stablecoexistenceindefinitelylarge}.

In this work, we develop an extremal graph theory argument to show that the critical coupling can take values only in the well‑defined interval $[\Delta^{-1},1],$ where $\Delta$ is the maximum degree of the network. We then focus on several random graph models and show that maximising the clustering while keeping the degrees fixed also tends to increase the critical coupling in networks sampled from these models, and the reverse also holds: optimising for critical coupling enhances clustering. We then turn to a real-world network inferred from herbaceous plant communities across Northern Eurasia~\cite{Scheifes2024} and show that both its clustering and its critical coupling are larger than what would be expected in the maximum‑entropy null model with an identical degree sequence. 
Taken together, these results suggest that the abundance of triangles in networks exceeding null-model predictions may reflect a selective bias toward systems that are more robust in the sense of asymptotic stability.

\section{Results}
\begin{figure}
    \centering
    \includegraphics[width=0.7\linewidth]{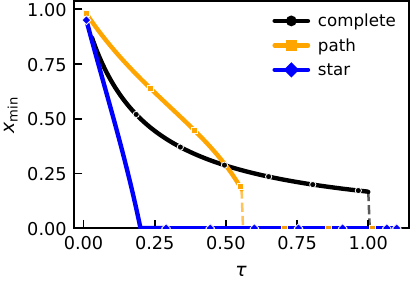}
    \caption{
    \textbf{Least abundant species vs. coupling strength.}
    The dependence of the minimum-abundance species value $x_{\min}$ at equilibrium on the coupling strength $\tau$ is shown for three graphs with 6 nodes: complete, path, and star. Abrupt jumps in the equilibrium branch appear as discontinuities, occurring at the critical coupling where the equilibrium loses feasibility.
    }
    \label{fig:phase_diagram}
\end{figure}
We consider a competitive LV system on a fixed, but possibly random, undirected network with a symmetric adjacency matrix $A$,
\begin{equation}
    \frac{\mathrm{d} x_i}{\mathrm{d}t} = x_i (1 - x_i) - \tau \sum_{j=1}^n A_{ij} x_i x_j, \label{eq:sys}
\end{equation}
where $\mathbf{x} \in \mathbb{R}^n$ is the vector of species abundances and the interaction strength $\tau \geq 0$ represents the competitive pressure.

In the absence of interactions between the nodes, meaning $\tau=0$, the system reduces to $n$ independent logistic equations, each with a carrying capacity of 1. In this case, the system will converge globally to an equilibrium state $\mathbf{x^*}$, where $x^*_i = 1$ for all $i=1,2,\dots,n$. The introduction of node interactions significantly alters the dynamics, with the system behaviour now depending on the underlying network structure.

The stability of system \eqref{eq:sys} is governed by its Jacobian,
\begin{equation}
    J(\mathbf{x})_{ik} = \left( 1 - 2 x_i - \tau \sum_{j=1}^n A_{ij} x_j \right) \mathds{1}_{\{ i = k \}} - \tau x_i \mathds{1}_{\{ A_{ik} = 1 \}}.
\end{equation}
Here, $\mathds{1}_{\{\cdot\}}$ denotes the indicator function (equal to $1$ if the condition holds and $0$ otherwise); thus $\mathds{1}_{\{ i = k \}}$ selects the diagonal terms, while $\mathds{1}_{\{ A_{ik} = 1 \}}$ selects the off-diagonal entries corresponding to edges in the network. At each value of $\tau$ for which it exists, the equilibrium fixed point with $x_i^* > 0$ for all $i = 1,2,\dots,n$ lies in the interior of the unit hypercube, is unique, and is given by
\begin{equation}
    \mathbf{x^*}(\tau) = (I + \tau A)^{-1} \mathbf{1}.
\end{equation}
If this feasible point is, in addition, linearly stable, then it is globally stable within the positive orthant~\cite{goh1977}. We refer to $\mathbf{x^*}(\tau)$ as a branch of feasible points. Along this branch, the Jacobian simplifies to
\begin{equation}
    J(\mathbf{x^*})_{ik}
    = -x_i \left( \mathds{1}_{\{ i = k \}} + \tau\, \mathds{1}_{\{ A_{ik} = 1 \}} \right).
    \label{eq:jac}
\end{equation}
Consider $\tau$ continuously increasing on the interval $[0,\tau_{\text{c}})$.
This gives rise to two distinct types of bifurcation points that depend on the network structure~\cite{mooij2024stablecoexistenceindefinitelylarge}:
\begin{enumerate}[label=(\arabic*)]
    \item A boundary transcritical bifurcation occurs when the branch of stable feasible points $\mathbf{x^*}(\tau)$ collides with one (or several) of the invariant hyperplanes \(H_i := \{\mathbf{x} \in \mathbb{R}^n : x_i = 0\}\).  
    
    \item A saddle node or pitchfork bifurcation occurs
    when the feasible equilibrium branch $\mathbf{x^*}(\tau)$ loses stability as the leading eigenvalue of the Jacobian crosses zero. By definition, this event happens when the branch is in the interior of the domain.
\end{enumerate}
 In both scenarios the feasible fixed point loses stability, and the corresponding equilibrium branch becomes confined to one of the boundary invariant hyperplanes $H_i$ (or to an intersection of several such hyperplanes).
This event is interpreted as the extinction of one (or several) species.

We define the \emph{critical coupling} $\tau_{\text{c}}$ as the smallest value of $\tau > 0$ for which at least one component of the fixed point $x_i^*(\tau)$ is zero, which corresponds to at least one extinction.
\begin{equation}
\tau_{\mathrm c} = \inf \bigl\{ \tau > 0 : \exists i \in \{1,\dots,n\} \text{ such that}\ x_i^*(\tau) = 0 \bigr\}.
\end{equation}
As illustrated in Figure~\ref{fig:phase_diagram}, the critical couplings correspond to discontinuities at the bifurcation points. 

In Section~\ref{sec:analytic} we derive lower and upper bounds for the critical coupling. Furthermore, we show in Section~\ref{sec:numerical} using optimisation algorithms, that the critical coupling and mean clustering are strongly dependent in several random graph models.
Finally, Section~\ref{sec:real_world_networks} shows a similar observation in real-world competition networks.

\subsection{Universal bounds for critical coupling}\label{sec:analytic}
\begin{figure}
    \centering
    \includegraphics[width=0.9\linewidth]{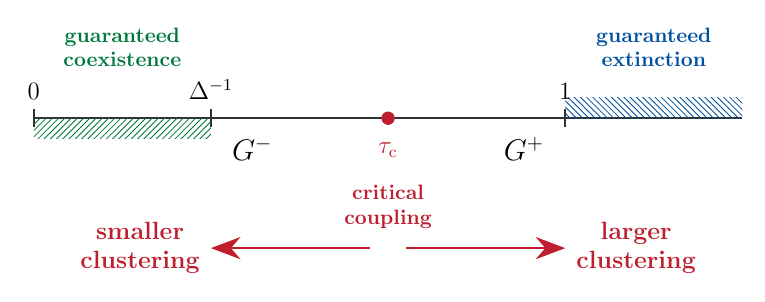}
    \caption{\textbf{Bounds for $\tau_{\mathrm{c}}$, the maximum coupling permitting coexistence.}
    Across all graphs $G$, the lower and upper bounds are attained at $G^-$ and $G^+$, respectively, where the definitions of $G^-$ and $G^+$ depend on the constraint imposed on $G$ -- either a fixed number of nodes or a fixed maximum degree, as explained in Theorem~\ref{th:general}.
    }
    \label{fig:tc}
\end{figure}
In this section, we derive lower and upper bounds for the critical coupling and show that $\tau_{\text{c}} \in [\Delta^{-1}, 1]$ for any graph with maximum degree $\Delta$. In particular, this shows that the lower bound on \(\tau_{\mathrm{c}}\) is independent of the network size, which shows that arbitrarily large biological systems can be stable.
We further show that for a fixed number of nodes $n$, the least robust structure is the star graph $S_n$, while the most robust is the complete graph $K_n$. 
The latter observation is also peculiar because for both $K_n$ and $S_n$ the maximum degree is the same, however the former maximises the number of triangles.

We formalize the resulting bounds in the following theorem, illustrated in Figure~\ref{fig:tc}.
\begin{theorem}\label{th:general}
    For every finite graph $G$ with maximum degree $\Delta>0$,
    \begin{equation}
        \tau_{\mathrm{c}}(G) \in \left[\Delta^{-1},\, 1\right],
        \label{eq:tc_bound}
    \end{equation}
    where the lower bound is attained whenever $G$ contains a component isomorphic to $S_{\Delta+1}$, while the upper bound is attained by graphs whose connected components are all complete graphs.

    Moreover, for every connected graph $G$ on $n\ge 2$ vertices, 
    \begin{equation}
        \tau_{\mathrm{c}}(G) \in \left[(n-1)^{-1},\, 1\right],
    \end{equation}
    and the lower bound is uniquely attained by the star graph $S_n$ and the upper bound by the complete graph $K_n$, for both of which the maximum degrees are $\Delta=n-1$.
\end{theorem}
 While Figure~\ref{fig:tc} illustrates the statements of Theorem~\ref{th:general}, informally, the strategy behind the proof is as follows. Coexistence breaks down when at least one component of the equilibrium vector $\mathbf{x}^*(\tau)$ reaches zero, so we analyse how $\mathbf{x}^*(\tau) = (I + \tau A)^{-1}\mathbf{1}$ evolves as $\tau$ increases. First, simple algebraic manipulations show that no component can vanish before $\tau = \Delta^{-1}$, yielding the lower bound $\tau_{\mathrm c} \ge \Delta^{-1}$ established in Lemma~\ref{lem:critical_coupling_lower}.  
Second, spectral arguments imply that the first possible bifurcation cannot be delayed beyond $\tau = 1$, giving the upper bound $\tau_{\mathrm c} \le 1$ proved in Lemma~\ref{lem:critical_coupling_upper}.  
Third, interpreting $(I+\tau A)^{-1}$ through its Neumann expansion shows that the lower bound is attained for star graphs, as shown in Lemma~\ref{lem:dmax}.  
Finally, an extremal result from spectral graph theory shows that the upper bound is attained for complete graphs, completing the proof via Lemma~\ref{lem:critical_coupling_upper}.

\begin{lemma}\label{lem:critical_coupling_lower}
    Let $G$ be a simple graph on $n$ vertices with maximum degree $\Delta$, and let $\tau_{\mathrm c}(G)$ denote the critical coupling. Then \[\Delta^{-1} \le \tau_{\mathrm c}(G) \le -\lambda_{\min}(A)^{-1}.\]
\end{lemma}
\begin{proof}
    Let $A$ be the adjacency matrix of $G$, and let $\rho(A)$ denote its spectral radius. The feasible branch is given by
    \begin{equation*}
        \mathbf{x^*}(\tau) = (I + \tau A)^{-1}\mathbf{1},
    \end{equation*}
    which is well-defined for $\tau <\rho(A)^{-1}$ and, in particular, for $\tau < \Delta^{-1}\le\rho(A)^{-1}$.
    
    For $\tau=0$ this branch is stable. As $\tau$ increases,
    either a saddle-node or pitchfork bifurcation occurs when the leading eigenvalue of the Jacobian
    \begin{equation*}
        J(\mathbf{x^*}) = -\diag(\mathbf{x^*})(I + \tau A)
    \end{equation*}
    crosses zero, which happens when
    $\tau = -\lambda_{\min}(A)^{-1}.$
    Hence no such bifurcation occurs for
    \begin{equation*}
        \tau < \Delta^{-1} \le \rho(A)^{-1} \le -\lambda_{\min}(A)^{-1}.
    \end{equation*}
    
    It remains to rule out a transcritical bifurcation for $\tau < \Delta^{-1}$. We write the fixed point componentwise,
    \begin{equation*}
        x_i^* = 1 - \tau \sum_{j} A_{ij} x_j^* \le 1,
    \end{equation*}
    and therefore
    \begin{equation*}
        x_i^* \ge 1 - \tau \deg(i) \ge 1 - \tau \Delta.
    \end{equation*}
    Thus
    \begin{equation*}
     x_i^*(\tau) \ge 1 - \tau \Delta \ge 0, \; \text{for all } i =1,\dots,n,
    \end{equation*} 
    and so $\mathbf{x^*}(\tau)$ stays away from the invariant hyperplanes for these values of $\tau$. 
\end{proof}

\begin{lemma}\label{lem:dmax}
    Among all connected graphs on $n$ vertices, the star graph $S_n$ uniquely minimises the critical coupling, with
    \begin{equation}
        \tau_{\mathrm c}(S_{n}) = (n-1)^{-1}.
    \end{equation}
\end{lemma}
\begin{proof}
    For the star graph, the smallest component of $\mathbf{x^*}$ occurs at the central node. For $\tau<\rho(A)^{-1}$, the resolvent admits the convergent Neumann expansion
    \begin{equation*}
        (I+\tau A)^{-1} = \sum_{i=0}^\infty (-\tau)^i A^i,
    \end{equation*}
    so that $\mathbf{x^*}(\tau) = \sum_{i=0}^\infty (-\tau)^i A^i\mathbf{1}$. The entry $(A^i)_{kj}$ equals the number of walks of length $i$ from $k$ to $j$, hence the row sum $(A^i\mathbf{1})_k$ equals the total number $W_i(k)$ of walks of length $i$ starting at $k$. The $k$-th component of $\mathbf{x^*}(\tau)$ therefore reads
    \begin{equation*}
        x_k^*(\tau)=\sum_{i=0}^{\infty}(-1)^i \tau^i W_i(k).
    \end{equation*}
    For the central vertex, denoted $1$, the equilibrium value is
    \begin{equation*}
        x_1^*(\tau) = \sum_{i=0}^\infty (-1)^i \tau^i (n-1)^{\lfloor (i+1)/2 \rfloor}.
    \end{equation*}
    Separating even and odd terms gives geometric sums
    \begin{align*}
        x_1^*(\tau)
        &= \sum_{j=0}^\infty \tau^{2j} (n-1)^j
        - \sum_{j=0}^\infty \tau^{2j+1} (n-1)^{j+1} \\
        &= \frac{1}{1-(n-1)\tau^2}
        - \frac{(n-1)\tau}{1-(n-1)\tau^2}
        = \frac{1-(n-1)\tau}{1-(n-1)\tau^2}.
    \end{align*}
    The smallest positive root is therefore
    \begin{equation*}
        \tau_{\mathrm c}(S_{n}) = (n-1)^{-1}.
    \end{equation*}

    To prove uniqueness, suppose that for some connected graph $G$ we have $x_i^*=0$ at $\tau=(n-1)^{-1}$. Then
    \begin{equation*}
        \sum_{j=1}^{n} A_{ij} x_j^* = \tau^{-1} = n-1.
    \end{equation*}
    Since each $x_j^* \le 1$, it follows that $x_j^* = 1$ and $\deg(i)=n-1$ for all neighbours $j$ of $i$. Thus $S_{n}$ centred at $i$ is a subgraph of $G$.
    
    For any neighbour $j$ of $i$,
    \begin{equation*}
        1 + \tau \sum_{k=1}^{n} A_{jk} x_k^* = 1,
    \end{equation*}
    so
    \begin{equation*}
        \sum_{k=1}^{n} A_{jk} x_k^* = 0.
    \end{equation*}
    Since $x_k^* \ge 0$ by continuity of $\mathbf{x^*}(\tau)$, every neighbour of $j$ must have value $0$. Hence $j$ is only connected to $i$, and therefore $G = S_{n}$.
\end{proof}

\begin{corollary}\label{cor:remark}
    For any simple graph $G$ with maximum degree $\Delta$, $\tau_{\mathrm{c}}(G) \ge \Delta^{-1}$. If one of the connected components of $G$ is a star $S_{\Delta+1}$, then equality holds.
\end{corollary}
\begin{proof}
    The lower bound follows from Lemma~\ref{lem:critical_coupling_lower}. 
    Recall that for a disconnected graph, the critical coupling equals the minimum of the critical couplings of its connected components. If $G$ has connected component $C \cong S_{\Delta+1}$, then 
    $\tau_{\mathrm{c}}(G) \le \tau_{\mathrm{c}}(S_{\Delta+1}) = \Delta^{-1}$. Combined with the general lower bound \(\tau_{\mathrm{c}}(G) \ge \Delta^{-1}\), we conclude that
    $
    \tau_{\mathrm{c}}(G) = \Delta^{-1}.
    $
\end{proof}
\begin{remark}
    Note that the sufficient condition in Corollary~\ref{cor:remark} is not necessary.
    Indeed, consider the double star on $7$ vertices: two hubs $i, \ell$ of degree $3$ sharing a common neighbour $j_1$, with two private leaves each (see Figure~\ref{fig:double_star}). Then $\Delta = 3$ and
    \begin{equation*}
        x_i^*(\tau) = x_\ell^*(\tau) = \frac{1-3\tau}{1-4\tau^2},
    \end{equation*}
    which vanishes at $\tau = \tfrac{1}{3} = \Delta^{-1}$. Thus $\tau_{\mathrm{c}} = \Delta^{-1}$, even though no component in the graph is a star $S_4$.
    \begin{figure}[h!]
        \centering
        \begin{tikzpicture}[
            hub/.style={circle, draw, fill=black!20, minimum size=6mm, inner sep=0pt},
            leaf/.style={circle, draw, fill=white, minimum size=6mm, inner sep=0pt}
        ]
            \node[hub] (i)  at (-1.5, 0) {$i$};
            \node[hub] (l)  at ( 1.5, 0) {$\ell$};
            
            \node[leaf] (j1) at (0, 0) {$j_1$};
            
            \node[leaf] (j2) at (-2.5,  1) {$j_2$};
            \node[leaf] (j3) at (-2.5, -1) {$j_3$};
            
            \node[leaf] (j4) at (2.5,  1) {$j_4$};
            \node[leaf] (j5) at (2.5, -1) {$j_5$};
            
            \draw (i) -- (j1);
            \draw (i) -- (j2);
            \draw (i) -- (j3);
            \draw (l) -- (j1);
            \draw (l) -- (j4);
            \draw (l) -- (j5);
        \end{tikzpicture}
        \caption{The double star with $\Delta = 3$. Shaded nodes are the hubs where $\mathbf{x^*}$ vanishes first.}
        \label{fig:double_star}
    \end{figure}
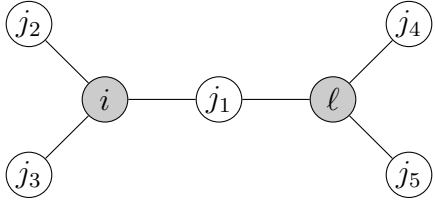
\end{remark}

\begin{lemma}\label{lem:critical_coupling_upper}
    Let $G$ be a simple graph on $n$ vertices, and let $\tau_{\mathrm c}(G)$ denote the critical coupling. Then
    \begin{equation}
        \tau_{\mathrm c}(G) \le 1.
    \end{equation}
    If $G$ is connected, the equality is attained if and only if $G\cong K_n$.
    If $G$ is not connected, the equality is attained if and only if every connected component of $G$ is a complete graph $K_{n_i}$ with $n_i \ge 2$.    
\end{lemma}
\begin{proof}
    Let $G$ be connected and $d$-regular. Then
    \begin{equation*}
        x_i^* = (1 + \tau d)^{-1} > 0,\; \text{for all } i =1,\dots,n,
    \end{equation*}
    for all $\tau > 0$, so no transcritical bifurcation occurs. In particular, let $G$ be a complete graph $K_n$, then $\lambda_{\min}(A) = -1$, hence
    \begin{equation*}
        \tau_{\mathrm c}(K_n) =-\frac{1}{\lambda_{\min}(A)}= 1.
    \end{equation*}
    However, $K_{n}$ uniquely maximises $\lambda_{\min}(A)$ among all connected simple graphs with $n$ vertices~\cite{brouwer2012spectra}. 
    Since $\tau_{\mathrm{c}}(G) \le -\lambda_{\min}(A)^{-1}$, as shown in Lemma~\ref{lem:critical_coupling_lower}, it follows that $\tau_{\mathrm{c}}(G) \le \tau_{\mathrm{c}}(K_n) = 1$.
    If $G$ is not connected, the same reasoning applies component-wise.
\end{proof}

\subsection{Network optimisation}\label{sec:numerical}
To quantify triangle density, we use the clustering coefficients
\begin{equation}
    C(G)=\frac{1}{n}\sum_{i=1}^{n} C_i,
    \qquad
    C_i=\frac{2T_i}{d_i(d_i-1)}, \; d_i\ge 2,
\end{equation}
where $T_i$ is the number of triangles incident to node $i$, so that $C_i$ measures the fraction of possible triangles around node $i$ that are realized (and $C_i=0$ when $d_i<2$).
The degree sequence of $G$ is denoted by $d_G=(d_1,\dots,d_n)$.

Theorem \ref{th:general} suggests that complete graphs, for which $C(G)=1$, are also the most robust in terms of having maximal $\tau_{\mathrm{c}}$. 
To investigate potential dependencies between the critical coupling $\tau_{\mathrm{c}}$ and the mean clustering $C$, we optimise the network structure with respect to these quantities.
Specifically, among all graphs with degree sequence $d_G$, we identify those that attain the smallest and largest possible mean clustering by solving
\begin{equation}
    \arg\min_{G'} C(G')
    \quad\text{and}\quad
    \arg\max_{G'} C(G')
    \qquad\text{s.t.}\quad d_{G'}=d_G.
\end{equation}

We will find these optimal network structures by setting up the double‑edge‑swap Markov chain that rewires edges $(i, j)$ and $(k, l)$  in two distinct ways: 1) perpendicular rewiring, to $(i, k)$ and $(j, l)$, and 2) crosswise rewiring to $(i, l)$ and $(j, k)$. These rewiring operations preserve the degree sequence of the network, ensuring that each node retains the same number of neighbours (Figure~\ref{fig:rewire_method}).
\begin{figure}[h!]
    \centering
    \includegraphics[width=0.5\linewidth]{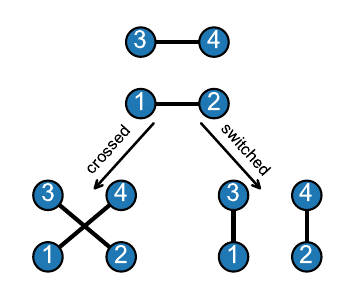}
    \caption{{\bf Illustration of the degree-preserving edge-swap (rewiring) operation.} Two edges connecting nodes $(1,2)$ and $(3,4)$ are replaced by either crossed or switched configurations, conserving the degree of each node while altering network topology.}
    \label{fig:rewire_method}
\end{figure}
More importantly, the double‑edge‑swap Markov chain is known to be irreducible \cite{greenhill2015switch}, which guarantees that any network with the same degree sequence is reachable.
In practice, the optimisation process may also spend a long time in various metastable states. Figure~\ref{fig:schematic} illustrates the optimisation procedure for a specific network that can be divided into disjoint cliques. We observe that optimisation of the mean clustering of the network increases the critical coupling, which aligns with the statement of Lemma~\ref{lem:critical_coupling_upper}.
\begin{figure}[h!]
    \centering
    \includegraphics[width=1.0\linewidth]{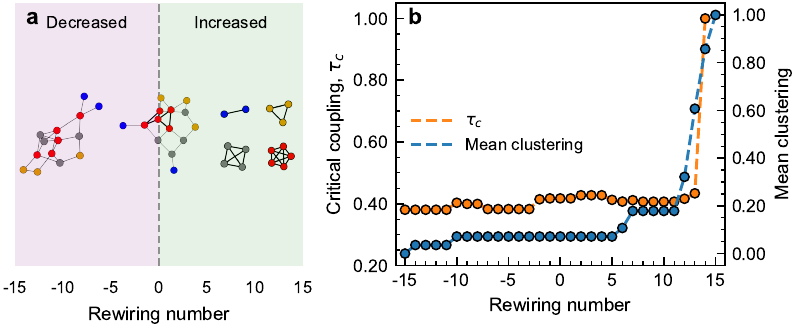}
    \caption{
    \textbf{Degree-preserving rewiring that increases clustering also increases the critical coupling.}
    Single network with $n=14$ nodes and $m=20$ edges.
    The rewiring number counts successive degree-preserving edge swaps relative to the initial network (0); positive values increase the mean clustering coefficient and negative values decrease it.
    \textbf{a}, Representative networks along the rewiring trajectory. The dashed vertical line marks the initial network. Purple (left) indicates reduced clustering and green (right) increased clustering.
    \textbf{b}, Critical coupling $\tau_{\text{c}}$ (left axis) and mean clustering coefficient (right axis) versus rewiring number.
    }
    \label{fig:schematic}
\end{figure}

To examine the full range of clustering and critical coupling, we apply the rewiring procedure to several random network models. 
Specifically, we consider random regular networks (uniform degree), configuration-model networks with degree sequences sampled from a Poisson distribution, geometric random networks, and networks generated by the Watts–Strogatz and Barabási–Albert models. 
All models are constructed to have the same mean degree.
For each model, we generate 100 networks with size 100. 
Within each model, all 100 realizations share the same degree sequence but differ in structure. 
For every network, we either optimise the mean clustering or the critical coupling of the network and record the measured values.

We observe that the mean clustering coefficient and critical coupling are positively correlated across all network models (Fig.~\ref{fig:clustering_coupling}). 
This relationship is strongest in random regular networks and weakest in Barabási–Albert networks, where clustering influences critical coupling only at low values. 
Variance peaks in Erdős–Rényi networks and remains small in random regular networks. 
Note that the optimisation procedure is important, as shown by the difference between obtained values in the first and second row. 
\begin{figure}
    \centering
    \includegraphics[width=1.0\linewidth]{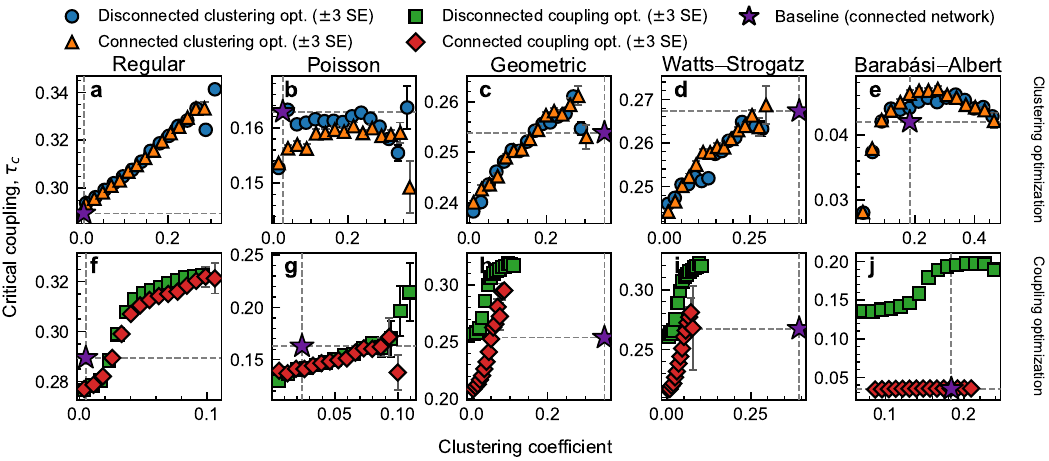}
    \caption{
    {\bf The critical coupling shows positive correlation with the mean clustering coefficient.}
    Each network is rewired to maximise or minimise the critical coupling $\tau_{\text{c}}$ or the clustering $C$ (1000 rewirings both for optimisation and minimisation).
    For each network model we report the critical coupling averaged over all graph realization obtained in all rewirings of $100$ starting configuration model networks.
    All networks have size $n=100$ and mean degree $\langle k\rangle \approx 4$. For every network type, all configurations share the same degree sequence.
    Circles and squares mark the unconstrained rewiring, which may disconnect the network; triangles and diamonds mark the rewiring constrained to connected graphs only. The two settings give the same positive trend in every model.
    {\bf a,f}, regular ($d=4$);
    {\bf b,g}, Erd\H{o}s--R\'enyi ($p=4/n$);
    {\bf c,h}, geometric ($r=2(\pi(n-1))^{-1/2}$);
    {\bf d,i}, Watts--Strogatz ($d=4$, $p=0.1$);
    {\bf e,j}, Barab\'asi--Albert ($m=2$).
    Error bars are $3$ standard error, while purple stars indicate the initial, unoptimised connected networks. 
    In the connected-clustering optimization panel for Barabási–Albert networks, the clustering changes, but the change is not visible at the scale of the plot.    }
    \label{fig:clustering_coupling}
\end{figure}

\subsection{Real-world networks}\label{sec:real_world_networks}
We apply our framework to real-world networks representing natural grassland ecosystems, where interspecific competition is primarily driven by resource limitation. 
We infer these networks from the dataset compiled by Scheifes et al.~\cite{Scheifes2024}, which combines vegetation-plot composition with measurements of nutrient concentrations --- nitrogen (N) and phosphorus (P) --- in the plants, in herbaceous plant communities across Northern Eurasia. 
We consider 872 plots classified into habitat types according to the EUNIS classification scheme~\cite{davies2004eunis}, summarized in Table~\ref{tab:habitats}. 
For each habitat, the set of all species observed across its plots defines the corresponding "species pool", from which we sample to construct competitive networks. 
Specifically, we draw random subsets of 100 species from each species pool and use them as the node set of the network.
Pairwise competition coefficients are computed from the niche overlap of species-specific resource-use strategies~\cite{haraldsson2023emerging}. 
We retain the strongest 10\% of edges.
The results do not depend on this percentage. Across the range of thresholds we tested, the realized networks have higher clustering and higher critical coupling than a null model with the same density (Methods, Fig.~\ref{fig:threshold_sensitivity}).
More details are provided in the Methods section.
To test whether the resulting network structure can be explained by degree heterogeneity alone, we compare each realized network to a configuration-model ensemble that preserves its degree sequence. 
For each environment, we generate 500 configuration-model realizations as simple networks (i.e., no self-loops and no multi-edges) with the same degree sequence as the realized network. 
We then compare the realized clustering coefficient difference $\Delta C$ and the critical coupling difference $\Delta \tau_{\text{c}}$ (Figure~\ref{fig:real_world}a) against the corresponding distributions under the null ensemble. 
In all environments, $\Delta C>0$ and $\Delta \tau_{\text{c}} >0$ for all $M=500$ configuration-model realizations.
\begin{table}[h!]
    \centering
    \caption{Habitat classification acronym,  group and type, with the number of species recorded per each habitat.}
    \small
    \setlength{\tabcolsep}{6pt}
    \renewcommand{\arraystretch}{1.15}
    \begin{tabularx}{\linewidth}{@{} l >{\raggedright\arraybackslash}X >{\raggedright\arraybackslash}X r @{}}
        \toprule
        \textbf{Var.} & \textbf{Habitat group} & \textbf{Habitat type} & \textbf{\# species} \\
        \midrule
        R1 & \multirow{3}{=}{Grasslands and lands dominated by forbs, mosses or lichens} & Dry grasslands & 128 \\
        R2 &  & Mesic grasslands & 230 \\
        R3 &  & Seasonally wet and wet grasslands & 303 \\
        \addlinespace[0.4em]
        N1 & Coastal habitats & Coastal dunes and sandy shores & 158 \\
        \addlinespace[0.4em]
        Q2 & \multirow{2}{=}{Wetlands} & Valley mires, poor fens and transition mires & 155 \\
        Q5 &  & Helophyte beds & 192 \\
        \addlinespace[0.4em]
        S9 & Heathlands, scrub and tundra & Riverine and fen scrub & 151 \\
        \bottomrule
    \end{tabularx}
    \label{tab:habitats}
\end{table}

\begin{figure}[h!]
    \centering
    \includegraphics[width=1.0\linewidth]{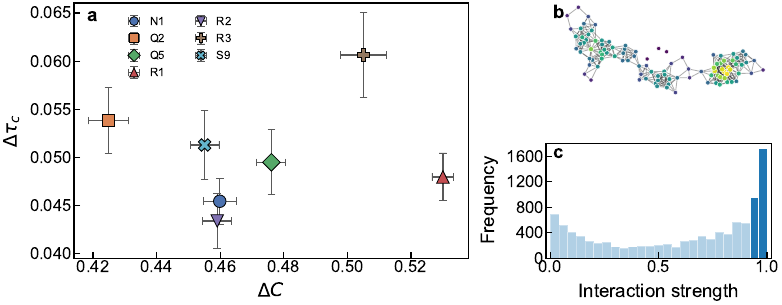}
    \caption{\textbf{Construction and properties of realized interaction networks.}
    \textbf{a,} Differences between real networks and degree-preserving configuration-model controls. 
    Horizontal axis $\Delta C = C_{\mathrm{real}} - C_{\mathrm{conf}}$ and vertical axis $\Delta \tau_{\text{c}} = \tau_{c,\mathrm{real}} - \tau_{c,\mathrm{conf}}$. $C_{\mathrm{real}}$ and $\tau_{c,\mathrm{real}}$ are measured on the real network. $C_{\mathrm{conf}}$ and $\tau_{c,\mathrm{conf}}$ are the corresponding configuration-model quantities averaged over 500 realizations. 
    Symbols denote environments as defined in Table~\ref{tab:habitats}. Results are averaged over 30 random subsets of 100 species from each environment, while error bars are one standard error.
    \textbf{b,} Example realized network of size $n=100$ after retaining the strongest 10\% of edges, where node size and colour are proportional to degree.
    \textbf{c,} Distribution of pairwise interaction strengths. The opaque region indicates interactions set to zero corresponding to absent edges in the realized network; 90\% of edges are removed.}
    \label{fig:real_world}
\end{figure}

\section{Discussion}
Triadic closure, quantified by the clustering coefficient $C$, stabilises competitive networks. 
We measure stability by the critical coupling, i.e. the competition strength at which exclusions first appear. 
To separate topology from degree effects, we compare against random-network controls that preserve the degree sequence while changing which pairs are linked. 
Under this constraint the direction is consistent across all simulations we report: higher clustering corresponds to larger critical coupling.
Real-world grassland networks follow the same trend, being more clustered and showing higher critical coupling than their degree-matched controls. 
This suggests that adaptation might be occurring in these networks.

Our analysis has been done for sparse and competitive networks, which are pervasive in nature, e.g.~\cite{connell1983}.
Extending the analysis to non-sparse networks, where species interaction is ubiquitous, could be an interesting follow-up.
Moreover, it still remains an open question if other types of systems, i.e. mutualistic, facilitative, or trophic, also show the same behaviour.
Analysing additional real-world networks, for example in gut microbiome~\cite{coyte2015, camacho2024}, can further strengthen our work.

If interactions are directed, the interaction matrix is no longer symmetric and our analysis no longer applies, since the stability bounds we derive rely on its eigenvalues being real. Recent random-matrix results show that, for such asymmetric networks, stability is governed by the sign pattern of the interactions~\cite{valigi2025eigenvalue}.

In our analysis, we allow for the networks to become disconnected.
This fragmentation can raise the critical coupling by containing extinction events in sub-networks, rather than allowing cascades to propagate through the full network. 
That mechanism is ecologically plausible: fragmented habitats and limited dispersal would localise competitive effects in much the same way.

Finally, let us mention that our results are potentially relevant to various disciplines outside of ecology. 
In economics, the Lotka-Volterra system is often considered. Examples are modelling economies and banking systems~\cite{comes2012banking, moran2019may}.
Also in computer science this system is used, for example in congestion control~\cite{pavlos2009, iguchi2005} or performing combinatorial optimisation~\cite{mooij2024findinglargeindependentsets}.

We also note that our analysis does not extend to directed graphs. For such systems, comparatively little is known about how specific structural features influence the associated dynamics. 
The first results in this direction concern spectral bounds, obtained from matrix norms of the adjacency matrix~\cite{LinLuTian2013}, were the degrees also appear to be important. 
Similarly,  directed cycles of arbitrary length~\cite{BrualdiLiu1995} (often referred to as feedback loops) also have effect on digraph spectra.
In addition, random models for directed networks, more fully developed in the literature (see, e.g., \cite{MetzNeriRogers2019}), can provide insight into which structural properties are relevant for stability.

Overall, we conclude that stability is not only dictated by node degrees. 
The arrangement of competitive links matters, and triadic closure in particular shifts the threshold at which competitive exclusion sets in.

\section{Methods}
\paragraph{Optimisation routine.}
We optimise by degree-preserving edge rewiring with a Metropolis-Hastings algorithm. Let $\mathcal{O}(G)$ denote the objective (mean clustering or critical coupling when increasing it, or its negative when decreasing it), and let $\delta=\mathcal{O}(G')-\mathcal{O}(G)$ be the change from a proposed rewiring $G \to G'$. We accept proposals with probability
\begin{align}
    P_{\mathrm{acc}}=
    \begin{cases}
    1, & \delta \ge 0,\\[4pt]
    \exp(\delta/T), & \delta < 0,
    \end{cases}
\end{align}
where $T>0$ is a temperature controlling acceptance of downhill moves (larger $T$ implies more exploration). This stochastic acceptance helps escape local minima.

For each network and for each direction (increase vs. decrease of mean clustering) we perform many rewirings.
At each step we sample candidate edge pairs uniformly at random and reject any move that would create self-loops or multi-edges, thereby preserving the degree sequence. For every network type we run multiple independent optimisations that share the same degree sequence, obtained by generating one instance of the model and re-randomising edges while preserving that sequence.
\\ \\
\paragraph{Global clustering coefficient.}
The mean clustering coefficient $C(G)$ averages a per-node quantity and weights all nodes equally regardless of degree. An alternative is the global clustering coefficient, defined as the fraction of paths of length two that are closed into a triangle,
\begin{equation}
    C_{\mathrm{global}}(G) = \frac{\text{number of closed paths of length two}}{\text{number of paths of length two}}.
\end{equation}
We repeated the rewiring procedure using $C_{\mathrm{global}}$ instead of $C$ as the clustering objective, restricting the rewiring to connected graphs.
We observe that the critical coupling and the global clustering coefficient are positively correlated across all network models and both optimisation directions (Fig.~\ref{fig:global_clustering}).
The trend matches the one obtained for the mean clustering coefficient, including for the Barabási–Albert model, where the degree distribution is most heterogeneous.
\begin{figure}[h!]
    \centering
    \includegraphics[width=1.0\linewidth]{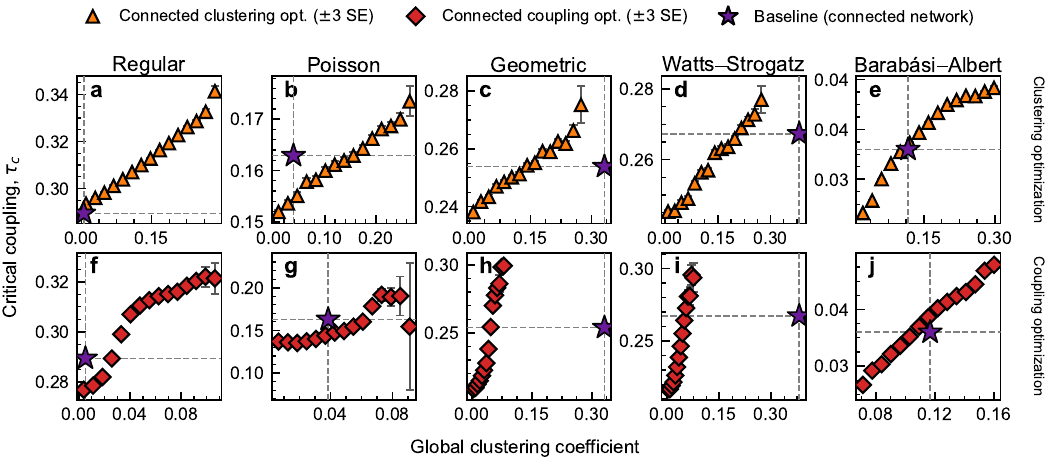}
    \caption{
    {\bf The critical coupling increases with the global clustering coefficient.}
    Each network is rewired to maximise or minimise the critical coupling $\tau_{\text{c}}$ or the global clustering coefficient $C_{\mathrm{global}}$ (1000 rewirings both for optimisation and minimisation), with the rewiring constrained to connected graphs.
    For each network model we report the global clustering coefficient and critical coupling across $100$ networks.
    All networks have size $n=100$ and mean degree $\langle k\rangle \approx 4$. For every network type, all configurations share the same degree sequence.
    Top row optimises clustering, bottom row optimises the critical coupling.
    {\bf a,f}, regular ($d=4$);
    {\bf b,g}, Erd\H{o}s--R\'enyi ($p=4/n$);
    {\bf c,h}, geometric ($r=2(\pi(n-1))^{-1/2}$);
    {\bf d,i}, Watts--Strogatz ($d=4$, $p=0.1$);
    {\bf e,j}, Barab\'asi--Albert ($m=2$).
    Error bars are three standard errors, while purple stars indicate the initial, unoptimised connected networks.
    }
    \label{fig:global_clustering}
\end{figure}
\\ \\
\paragraph{Calculation of the critical coupling.}
To determine the critical coupling, we follow the approach of~\cite{mooij2024findinglargeindependentsets}, and apply Newton's method to the function $F(\tau) = \prod_{i=1}^{n} x_i^*(\tau),$
where
$\mathbf{x}^*(\tau) = (I + \tau A)^{-1}\mathbf{1}.$
This yields the iteration
\begin{equation}
    \tau_{k+1}= \tau_k-\frac{\prod\limits_{i=1}^{n} x_i^*(\tau_k)}{\sum\limits_{j=1}^{n}
        \left[-\bigl(M^{-1}AM^{-1}\mathbf{1}\bigr)_j\right]\prod\limits_{i \neq j} x_i^*(\tau_k)},
\end{equation}
where $M = I + \tau A$, leading to the bifurcation value $\tau_{\text{c}} = \lim_{k \to \infty} \tau_k$. The equilibrium point $\mathbf{x}^* > 0$ can be obtained by inverting the matrix $M$. This inverse exists for all coupling strengths $\tau$ smaller than $|\lambda_{\min}(A)|^{-1}$,
where $\lambda_{\min}(A)$ is the smallest eigenvalue of $A$. When no transcritical bifurcation occurs first, the inverse ceases to exist at $\tau = |\lambda_{\min}(A)|^{-1}$ and the system undergoes a pitchfork bifurcation.

\paragraph{Grassland competitive interaction networks}
Our aim is to obtain a data-based approximation of competitive structure within each environment using the nutrient-limitation information available in the dataset. The common ecological assumption is that species are more likely to compete strongly when their environmental preferences ("niches") are more similar, i.e. their niches overlap more \cite{wandrag2019quantifying,chesson2000mechanisms}. To implement this, we defined their niches based on the nutrient concentrations measured in the plots where they occurred. If they tend to occur under similar nutrient concentration conditions, their competition is stronger. In particular,  we use the ratio of nitrogen and phosphorus $N:P$ as the relevant niche to define competition, as it has been identified as a key driver of the dynamics of these grasslands \cite{fay2015grassland,gusewell2004n,olde2003species}. For each species $i$, we collect the set of $N:P$ values from plots in which the species is present, where presence is defined as abundance $>0$. We then calculate the mean and standard deviation of these values, denoted $(\mu_i,\sigma_i)$. Finally, we assume that the niche of the species $i$  along the $N:P$ axis is a normal distribution $f_i(x)$ with mean $\mu_i$ and standard deviation $\sigma_i$ where $x := N:P$. Following a common approach in ecology \cite{haraldsson2023emerging,mason2005functional}, pairwise interaction strength is given by the niche overlap, quantified by the Pianka index \cite{pianka1973structure}.

We retain only the strongest inferred interactions, removing the weakest 90\% of off-diagonal overlap values and keeping the strongest 10\% (Figure~\ref{fig:real_world}a). 
We performed a sensitivity analysis on the exact percentage of edges removed, see below.
After thresholding, retained overlaps are treated as edges and removed overlaps as absent edges. 
The realized network is represented by an unweighted adjacency matrix $A$ obtained by binarising the thresholded overlaps; self-interactions are excluded ($A_{ii}=0$). An example realized network with $n=100$ is shown in Figure~\ref{fig:real_world}b.

\paragraph{Sensitivity analysis of real-world networks.} 
To assess if our results were sensitive to the exact percentage of removed edges, we repeated the analysis for the percentages $75\%,$ $80\%,$ $85\%,$ $90\%,$ $95\%$ for three habitats (R1, Q5 and S9). 
Each sampled network has 100 nodes; for each habitat/percentage we generated 50 sampled networks and, for each, 100 configuration-model networks with matched degree sequence.
We observe that clustering and the critical coupling are closer to the values predicted by the configuration model. However, they remain distinctly higher overall~(Fig.~\ref{fig:real_world_sensitivity}).

The number of retained edges sets the density of the realized network, and density by itself affects the clustering coefficient. To control for this, we compare each realized network against a null with the same number of edges and the same degree sequence, and vary the retained fraction from a few percent up to 50\%. For each fraction we measure the mean clustering coefficient, the global clustering coefficient, and the critical coupling. We observe that the realized networks exceed the null in all three quantities at every threshold (Fig.~\ref{fig:threshold_sensitivity}). The difference is already present below 10\%, so the choice of threshold does not drive the result.

\begin{figure}[h!]
    \centering
    \includegraphics[width=\linewidth]{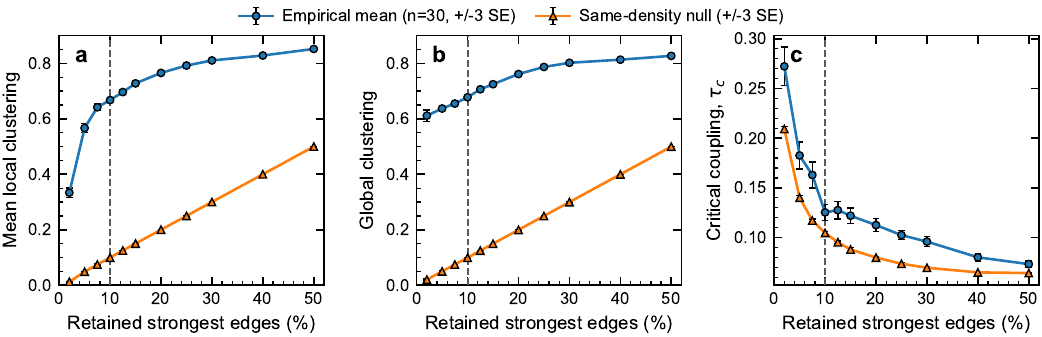}
    \caption{
    {\bf Clustering and stability against a same-density null across thresholds for the environment R1 (dry grasslands).}
    Mean local clustering ({\bf a}), global clustering ({\bf b}), and critical coupling $\tau_{\mathrm{c}}$ ({\bf c}) as the fraction of strongest retained edges varies from a few percent to 50\%.
    Blue: empirical networks (mean over $n=30$ sampled networks, $3$ standard errors). Orange: a null with the same edge count and degree sequence at each threshold.
    The dashed line marks the 10\% cutoff used in the main text.
    The empirical curve lies above the null at every threshold for all three quantities.
    }
    \label{fig:threshold_sensitivity}
\end{figure}

\begin{figure}[h!]
    \centering
    \includegraphics[width=0.65\linewidth]{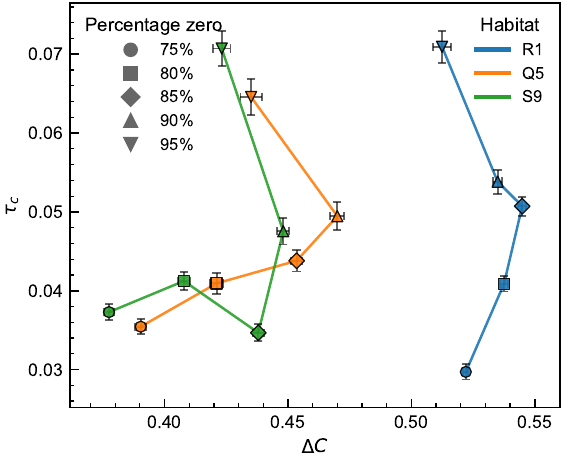}
    \caption{
    {\bf Sensitivity of the analysis to the percentage of edges that are removed across habitats.} 
    Points show percentage means for each habitat (blue/orange/green), with error bars indicating one standard error. 
    Lines connect percentages within each habitat; marker shapes denote the percentage of edges removed.
    Each sampled network has 100 nodes; for each habitat/percentage we generated 50 sampled networks and, for each, 100 configuration-model networks with matched degree sequence. 
    }
    \label{fig:real_world_sensitivity}
\end{figure}

\section{Code availability}
The code for this paper is available at\\ \url{https://github.com/NiekMooij/An-adaptation-mechanism-regulating-triadic-closure}.

\section{Data availability}
No data is collected for this manuscript.


\bibliographystyle{naturemag}

\section{Funding}
This work was supported by the Italian National Biodiversity Future Center (NBFC): National Recovery and Resilience Plan (NRRP), Mission 4 Component 2 Investment 1.4 of the Italian Ministry of University and Research; funded by the EU - NextGenerationEU (Project code CN 00000033). MB also acknowledges funds from the SENTINEL PRIN project – Call 2022, Grant No. 2022CM4F3X.
This work was supported by the following NWO programs: Talentprogramme, project number VI.C.202.081; and research program VIDI, project number VI.Vidi.213.108. 
Furthermore, the project acknowledges funds from the SENTINEL PRIN project – Call 2022, Grant No. 2022CM4F3X by the Italian Ministry of University and Research (MUR);  from the DiviN-P project under the 2021-2022 BiodivProtect joint call, co-funded by the European Commission (GA No. 101052342) and MUR (CUP:B53C23001040006); and from the Italian National Biodiversity Future Center (NBFC): National Recovery and Resilience Plan (NRRP), Mission 4 Component 2 Investment 1.4 of MUR, funded by the EU - NextGenerationEU (CN 00000033).

\section{Acknowledgements}
The authors are grateful to Annegreet Veken for discussion that helped define the real-world networks and for supplying the data used to generate these networks. 
The authors are grateful to Jochem Hoogendijk, Mike de Vries, and Ruben Hendriks for helpful discussions. 
MNM gratefully acknowledges support from Complex Systems Fund, with special thanks to Peter Koeze. 
MB acknowledges funds from the SENTINEL PRIN project – Call 2022, Grant No. 2022CM4F3X by the Italian Ministry of University and Research (MUR);  from the DiviN-P project under the 2021-2022 BiodivProtect joint call, co-funded by the European Commission (GA No. 101052342) and MUR (CUP:B53C23001040006); and from the Italian National Biodiversity Future Center (NBFC): National Recovery and Resilience Plan (NRRP), Mission 4 Component 2 Investment 1.4 of MUR, funded by the EU - NextGenerationEU (CN 00000033).
AvdH gratefully acknowledges support from NWO under the NWO Talentprogramme, project number VI.C.202.081.
IK gratefully acknowledges support from Netherlands Research Organisation (NWO), research program VIDI, project number VI.Vidi.213.108. 

\section{Competing interests}
The authors declare no competing interests.

\section{Materials \& Correspondence}
Please contact MNM for correspondence.

\end{document}